\documentclass[12pt, letterpaper]{article}
\usepackage[utf8]{inputenc}
\usepackage{natbib}
\usepackage{amsmath}
\usepackage{amsthm}
\usepackage{graphicx}
\usepackage{subfig}
\graphicspath{{figures/}}
\usepackage{float}
\usepackage{xcolor}

\newtheorem{theorem}{Theorem}[section]
\newtheorem{assumption}{Assumption}
\newtheorem{proposition}{Proposition}
\newtheorem{lemma}{Lemma}

\title{Ridge regularization for Mean Squared Error Reduction in Regression with Weak Instruments}
\author{
	Karthik Rajkumar\thanks{Department of economics, Stanford University. krajkumar@stanford.edu. I am very grateful to my advisor, Guido Imbens, for advising me through this project with patience and insight.}
}

\begin{document}
	
	\maketitle
	
	\begin{abstract}
		In this paper, I show that classic two-stage least squares (2SLS) estimates are highly unstable with weak instruments. I propose a ridge estimator (ridge IV) and show that it is asymptotically normal even with weak instruments, whereas 2SLS is severely distorted and unbounded. I motivate the ridge IV estimator as a convex optimization problem with a GMM objective function and an $L_2$ penalty. I show that ridge IV leads to sizable mean squared error reductions theoretically and validate these results in a simulation study inspired by data designs of papers published in the American Economic Review.
		
	\end{abstract}

\section{Introduction}
Instrumental variables are widely used in applied economics and other social sciences to establish causal relationships in the absence of experimental variation. Under standard assumptions, instrumental variable (IV) regression estimators are unbiased and asymptotically consistent. However, when the first-stage, i.e. the relationship between instrumental variables and the independent variable at hand, is weak, inference with IV regression is distorted and a vast literature has emerged showing this. In particular, size is not controlled in finite sample and several finite-sample corrections of the IV estimator have been proposed to tackle this. These methods attempt to remove bias in small sample in a way that washes out as sample sizes grow. \\

Many of these methods, however, do so in a static way, which sometimes leads to overcorrection or even having no effect at all. We propose a ridge estimator for IV regression that attempts to alleviate the bias in a way that is tunable to the data. We motivate the estimator by returning to the classical perspective of IV regression as a ratio of two estimands (in the case of a single, just-identified IV), and showing that our approach is equivalent to stabilizing the denominator away from 0, thus avoiding the division-by-zero problem.\\


The paper is organized as follows. Section \ref{sec-lit} provides background on the weak instrument problem and IV outlier sensitivity problem. Section \ref{sec-model} introduces the data model we use throughout the paper. Section \ref{sec-ridge} motivates and defines the ridge IV estimator and shows full preservation of efficiency under ``large'' parameters. Section \ref{sec-staiger-stock} takes a local asymptotic approach to modeling weak instruments and shows using theory how ridge IV leads to drastic reductions in mean squared error. Section \ref{sec-interpretation} interpets the ridge IV estimator as the solution to a convex optimization problem with a GMM objective function with an $L_2$ penalty on the coefficient. Section \ref{sec-results} uses a simulation study whose parameters are tuned to be consistent with the data generating processes implied in a sample of papers published in the American Economic Review and shows how ridge IV can lower MSE in practice. Section \ref{sec-disc} concludes. \\

\section{Literature} \label{sec-lit}

There is a long literature in econometrics studying weak instruments. \cite{young2018consistency} raises concerns on the quality of inference with instrumental variables. Specifically, the problems cited are weak or irrelevant instruments, non-iid error processes, and distortion in inference from one or two observations. A major claim in that paper is that IV estimates have larger MSE than OLS (OLS here referring to the structural equation of directly regressing the dependent variable on the endogenous variable, bypassing any instruments). In situations arising in practice, one is unable to tell the two estimates apart in the sense that IV confidence intervals generally include OLS anyway, and implies a preference for OLS. This calls into question the utility of traditional econometric tests for endogeneity, such as the Hausman test (e.g. \cite{hausman1978specification}).\\

 \cite{andrews2018weak} provides a recent survey on weak instrument diagnostics, and inferential methods that are robust to weak instruments. They focus on the case of heteroskedastic and possibly non-iid errors, and make a case for the adoption of the robust F-statistic proposed by \cite{olea2013robust} in the case of a single endogenous regressor. \cite{young2018consistency} counters this, claiming that such pre-tests for detecting weak instruments do little to accurately diagnose the problem in practice.\\
 
 \subsection{The weak instrument problem} 
 The primary problem with weak instruments is that they are biased towards OLS. This gives tests the wrong size, and leads to misleading inference. By far, the most popular case in the literature appears to be the just-identified case with a single instrument and heteroskedasticity \citep{andrews2018weak}. To combat this, the most common approach used in the literature is some form of the following two-stage process:
 \begin{enumerate}
 	\item $F \geq 10$. Here instruments are not weak and regular 2SLS inference is used.
 	\item $F < 10$. Various ``weak instrument robust'' methods are used. 
 \end{enumerate}

What are some of these weak instrument robust methods? \cite{hirano2015location} show there does not exist an unbiased or asymptotically unbiased IV estimator. Anderson-Rubin confidence sets \citep{anderson1949estimation} are optimal in just-identified case with single instrument, even with heteroskedasticity. In over-identified models, conditional LR (CLR) test is a good test for homoskedastic settings as it is fully robust to weak instruments. \cite{andrews2017unbiased} provide methods that are unbiased when sign of first-stage coefficient is known a priori.\\

 For the scope of this paper, we address the large MSE critique of 2SLS. We provide a novel estimator, guiding theory, as well as simulation evidence for the utility of this estimator. 

\section{Model} \label{sec-model}
We begin with a a just-identified setting with one instrument. Our data takes the form 
\begin{align}
\label{model}
Y_{i}=\beta_{0}+\beta_{1}D_{i}+\epsilon_{i} \\
D_{i}=\pi_{0}+\pi_{1}Z_{i}+\eta_{i}. \notag
\end{align}

Each datapoint is iid and the instruments are exogenous, i.e. $Z_{i}$ are independent of $\eta_{i}$ and $\epsilon_{i}$. In this notation, $Y_i$ is the outcome variable of interest, and $D_i$ is the endogenous variable whose effect of the outcome one is interested in.\\

The main contention we want to address with our paper is that 2SLS is more sensitive because it is a ratio of two things and its p-value does not account for this stochasticity of the denominator. To explain this simply, in the just-identified case, our 2SLS estimator is 
\[
\hat{\beta}_{2SLS}=\frac{\sum_{i}(Y_{i}-\bar{Y})(Z_{i}-\bar{Z})}{\sum_{i}(D_{i}-\bar{D})(Z_{i}-\bar{Z})},
\]
which basically is
\[
=\frac{\hat{\text{Cov}}[Y_{i},Z_{i}]}{\hat{\text{Cov}}[D_{i},Z_{i}]} := \frac{a_n}{b_n},
\]
where $b_n \rightarrow 0$. Essentially, this is a a division-by-zero problem.

\section{Ridge estimation} \label{sec-ridge}

To address the weak instrument problem, we propose adding a bias to the denominator of the 2SLS, in the case of just-identified single instrument. This turns the estimator into
\[
\frac{a_n}{b_n + \lambda},
\]
where $\lambda$ is a tuning parameter, appropriately chosen. This serves the dual objectives of stabilizing the denominator of the IV estimator and also shrinking the coefficient toward zero, thus controlling type-1 error. This naturally motivates a ridge estimator in the case of just-identified with multiple instruments too:
\begin{equation}
(Z^{\prime}D+\lambda I)^{-1}Z^{\prime}Y.
\end{equation}
Similarly for an over-identified setting, the ridge estimator modifies to
\begin{equation}
\label{defn-overidentified-ridgeIV}
(D^{\prime}Z(Z^{\prime}Z)^{-1}Z^{\prime}D+\lambda I)^{-1}(D^{\prime}Z(Z^{\prime}Z)^{-1}Z^{\prime}Y).
\end{equation}

Suppose we allowed that the ridge estimator allowed for the penalty parameter to vary with sample size. That is, $\lambda$ is indexed by $n$ to give us a full sequence of penalty parameters. In our specific example from equation \ref{model}, this would mean our estimator is 
\begin{align}
\hat{\beta}_{\text{ridge}} & = \frac{\sum_{i}(Y_{i}-\bar{Y})(Z_{i}-\bar{Z})}{\sum_{i}(D_{i}-\bar{D})(Z_{i}-\bar{Z})+\lambda_{n}} \notag\\
& = \frac{\hat{\text{Cov}}[Y_{i},Z_{i}]}{\hat{\text{Cov}}[D_{i},Z_{i}]+\frac{\lambda_{n}}{n}}. 
\label{eqn-ridge-defn}
\end{align}

Before we dive into the asymptotic properties of the ridge estimator, we make a note on the sampling process we assume for the data. 

\subsection{Sampling assumptions on the data}
There are many assumptions that may be used for the data generating process. Suppose we condition on the instruments, $Z_i$, treating them as ``constants'' and only assuming the existence of probability limits of their moments. Then the classic OLS asymptotic results hold. What does that mean? Let us formalize this idea.

Specifically 

\begin{assumption}
	\label{assumption-constant}
	The instruments $Z_i$ are ``constants'' with respect to the error shocks. This is achieved by conditioning on the instruments in hand. Further, under this sampling scheme, we assume $\frac{1}{n} \sum_{i}Z_i^2 \rightarrow 1$.
\end{assumption}

Now consider the first stage equation:
\[
D_i = \pi_1 Z_i + \eta_i.
\]
We ignore constants for the time being.

\begin{lemma}[CLT for OLS---Version 1]
	Under Assumption \ref{assumption-constant}, the CLT for OLS is 
	\[
	\sqrt{n} \left( \hat{\pi_1} - \pi_1 \right) \stackrel{d}{\rightarrow} \mathcal{N}\left( 0, \sigma_{\eta}^2 \right).
	\]
\end{lemma}

\begin{proof}
	This is easily proven using the Liapounov CLT. Specifically, we have 
	\[
	D_i Z_i = \pi_1 Z_i^2 + Z_i \eta_i \sim \left( \pi_1 Z_i^2, Z_i^2 \sigma_{\eta}^2 \right).
	\]
	Thus, Liapounov CLT tells us that
	\[
	\frac{\sum_i D_i Z_i - \pi_1 Z_i^2}{\sqrt{\sum_i Z_i^2 \sigma_{\eta}^2}} \stackrel{d}{\rightarrow} \mathcal{N}(0, 1).
	\]
	Using the second moment condition assumed from the $Z_i$ sampling, we get the desired result.
\end{proof}

However, this is not the result one get when assuming $Z_i$ is fully stochastic!

\begin{assumption}
	\label{assumption-stochastic}
	Suppose the instruments $Z_i$ are fully stochastic. Without loss of generality, we assume unit variance and also bounded fourth moment, $m_4$.
\end{assumption}

\begin{lemma}[CLT for OLS---Version 2]
	Under Assumption \ref{assumption-stochastic}, the CLT for OLS is 
	\[
	\sqrt{n} \left( \hat{\pi_1} - \pi_1 \right) \stackrel{d}{\rightarrow} \mathcal{N}\left( 0, \pi_1^2 (m_4 - 1) + \sigma_{\eta}^2\right).
	\]
\end{lemma}

\begin{proof}
	When $Z_i$ is also stochastic, the random variable $D_iZ_i$ looks like
	\[
	D_iZ_i = \pi_1 Z_i^2 + Z_i \eta_i \sim \left(\pi_1, \pi_1^2 (m_4 - 1) + \sigma_{\eta}^2\right).
	\]
	
	Here we may apply the classic Lindeberg-Levy CLT to get
	\[
	\frac{\sqrt{n} \left(\frac{\sum_i D_iZ_i}{n} - \pi_1\right)}{\sqrt{\pi_1^2 (m_4 - 1) + \sigma_{\eta}^2}} \stackrel{d}{\rightarrow} \mathcal{N} \left(0, 1\right).
	\]
	Rearranging terms gives us the desired result. 
\end{proof}

Observe that the results under Assumptions \ref{assumption-constant} and \ref{assumption-stochastic} are different! Indeed the variance under full stochasticity is larger than the previous case, because there is more randomness to account for.\\

\begin{theorem}[Consistency of the ridge estimator]
	Let $\hat{\beta}_{\text{ridge}}$ be the ridge estimator as defined in Equation (\ref{eqn-ridge-defn}). Suppose the instrument is not totally irrelevant, i.e. $\pi_1 \neq 0$. Then under either Assumption \ref{assumption-constant} or Assumption \ref{assumption-stochastic}, and for a sequence of penalty parameters $\lambda_{n} = o(n)$, the estimator is consistent. That is, 
	\[
	\hat{\beta}_{\text{ridge}} \stackrel{p}{\rightarrow}\beta_1.
	\]
\end{theorem} 

\begin{proof}
We can rewrite the structural equation in the data generating process as 
\begin{align}
Y_{i} = \beta_{0}+\beta_{1}\left(\pi_{0}+\pi_{1}Z_{i}+\eta_{i}\right)+\epsilon_{i} \notag \\
= \left(\beta_{0}+\beta_{1}\pi_{0}\right)+\beta_{1}\pi_{1}Z_{i}+\left(\beta_{1}\eta_{i}+\epsilon_{i}\right). \notag
\end{align}

This is just the reduced form equation. Without loss of generality, we have assumed $\sum_iZ_i^2/n\rightarrow 1$ (or $\text{Var}[Z_i] = 1$, if using Assumption \ref{assumption-stochastic}). Then it is clear that 
\[
\hat{\text{Cov}}[Y_{i},Z_{i}] \stackrel{p}{\rightarrow} \beta_{1}\pi_{1}.
\]
Similarly, the first stage gives us  
\[
\hat{\text{Cov}}[D_{i},Z_{i}] \stackrel{p}{\rightarrow} \pi_{1}.
\]
Since $\pi_1 \neq0$, and we have $\frac{\lambda_n}{n} \rightarrow 0$, convergence in probability allows us to take the ratio of these two estimators, and we get our result.

\end{proof}

The natural next question is, what is the asymptotic distribution of the proposed estimator? We tackle this in the next theorem. 

\begin{theorem}[Asymptotic normality of the ridge estimator---Version 1]
	\label{thm-ridge-normality1}
	Suppose $\lambda_n = o(\sqrt{n})$. Then the ridge estimator, as defined in Equation (\ref{eqn-ridge-defn}), is asymptotically normal. That is,
	\[
	\sqrt{n}\left(\hat{\beta}_{\text{ridge}} - \beta_1 \right) \rightarrow \mathcal{N}(0, V_{\text{ridge}}).
	\]
	Further, under Assumption \ref{assumption-constant}, $V_\text{ridge} = \sigma_{\epsilon}^2 / \pi_1^2$.
\end{theorem}

\begin{proof}
	First, we examine the reduced form regression. We understand that from the Central Limit Theorem of OLS in Lemma \ref{assumption-constant}, we have 
	\[
	\sqrt{n}\left(\hat{\text{Cov}}[Y_{i},Z_{i}]-\beta_{1}\pi_{1}\right)\rightarrow\mathcal{N}\left(0,\sigma_{\text{red}}^{2}\right),
	\]
	where $\sigma_{\text{red}}^{2}$ is the homoskedastic error variance of the reduced form regression. That is, it is the variance of the $\left(\beta_{1}\eta_{i}+\epsilon_{i}\right)$ residual term. Similarly, analyzing the first-stage regression, we have 
	
	\[
	\sqrt{n}\left(\hat{\text{Cov}}[D_{i},Z_{i}]-\pi_{1}\right)\rightarrow\mathcal{N}\left(0,\sigma_{\eta}^{2}\right),
	\]
	for $\sigma_{\eta}^{2}$, the variance of $\eta_{i}$, the errors in the first stage.\\
	
	Next, we have $\text{Cov}[Y_i Z_i, D_i Z_i] = Z_i^2 \text{Cov}[Y_i, D_i] = Z_i^2\left( \beta_1 \sigma_{\eta}^2 + \sigma_{\epsilon \eta}^2 \right)$, where $\sigma_{\epsilon \eta}^2$ is the covariance between $\epsilon_{i}$ and $\eta_i$. Putting these results together, we get the multivariate result
	\[
	\sqrt{n}\left(\left(\begin{array}{c}
	\hat{\text{Cov}}[Y_{i},Z_{i}]\\
	\hat{\text{Cov}}[D_{i},Z_{i}]
	\end{array}\right)-\left(\begin{array}{c}
	\beta_{1}\pi_{1}\\
	\pi_{1}
	\end{array}\right)\right)\rightarrow\mathcal{N}\left(\left(\begin{array}{c}
	0\\
	0
	\end{array}\right),\left[\begin{array}{cc}
	\sigma_{\text{red}}^{2} & \beta_1 \sigma_{\eta}^2 + \sigma_{\epsilon \eta}^2\\
	\beta_1 \sigma_{\eta}^2 + \sigma_{\epsilon \eta}^2 & \sigma_{\eta}^{2}
	\end{array}\right]\right).
	\]
	
	Call this covariance matrix $\Sigma$. Because $\lambda_n = o(\sqrt{n})$, we can use Slutsky's theorem to get the same asymptotic distribution:
	\[
	\sqrt{n}\left(\left(\begin{array}{c}
	\hat{\text{Cov}}[Y_{i},Z_{i}]\\
	\hat{\text{Cov}}[D_{i},Z_{i}] + \frac{\lambda_n}{n}
	\end{array}\right)-\left(\begin{array}{c}
	\beta_{1}\pi_{1}\\
	\pi_{1}
	\end{array}\right)\right)\rightarrow\mathcal{N}\left(\left(\begin{array}{c}
	0\\
	0
	\end{array}\right),\Sigma \right).
	\]
	
	This is the asymptotic distribution of a bivariate estimator. We note that $\hat{\beta}_{\text{ridge}}$ is simply the ratio of the first and the second elements of this bivariate estimator. \\
	
	Then to obtain the distribution of the ridge estimator, we can use the multivariate delta method here. It says the ratio estimator converges to the ratio of individual probability limits. Consider the bivariate function $h(x, y) = x/y$. Its gradient is given by $ \nabla h(x,y) = \left(1/y \text{, } -x/y^2 \right)^T$. So the asymptotic variance of the ridge estimator is 
	\[
	\nabla h(\beta_1 \pi_1, \pi_1)^T \times \Sigma \times \nabla h(\beta_1 \pi_1, \pi_1)
	\]
	\[
	= \left( \begin{array}{cc}
	\frac{1}{\pi_1} & -\frac{\beta_1}{\pi_1}
	\end{array} \right) 
	\left[\begin{array}{cc}
	\sigma_{\text{red}}^{2} & \beta_1 \sigma_{\eta}^2 + \sigma_{\epsilon \eta}^2\\
	\beta_1 \sigma_{\eta}^2 + \sigma_{\epsilon \eta}^2 & \sigma_{\eta}^{2}
	\end{array}\right]
	\left( \begin{array}{c}
	1/\pi_1 \\
	-\beta_1/ \pi_1
	\end{array} \right).
	\]
	
	Here $\sigma_{\text{red}}^2 = \beta_1^2 \sigma_{\eta}^2 + \sigma_{\epsilon}^2 + 2\beta_1 \sigma_{\epsilon \eta}^2$. So the above is 
	\[
	= \left( \begin{array}{cc}
	\frac{1}{\pi_1} & -\frac{\beta_1}{\pi_1}
	\end{array} \right) 
	\left[\begin{array}{cc}
	\beta_1^2 \sigma_{\eta}^2 + \sigma_{\epsilon}^2 + 2\beta_1 \sigma_{\epsilon \eta}^2 & 
	\beta_1 \sigma_{\eta}^2 + \sigma_{\epsilon \eta}^2\\
	\beta_1 \sigma_{\eta}^2 + \sigma_{\epsilon \eta}^2 & \sigma_{\eta}^{2}
	\end{array}\right]
	\left( \begin{array}{c}
	1/\pi_1 \\
	-\beta_1/ \pi_1
	\end{array} \right),
	\]
	which simplifies to 
	\[
	= \frac{1}{\pi_1^2} \left[ \begin{array}{cc}
	\sigma_{\epsilon}^2 + \beta_1 \sigma_{\epsilon \eta}^2 &
	\sigma_{\epsilon \eta}^2
	\end{array} \right]
	\left( \begin{array}{c}
	1 \\
	-\beta_1
	\end{array} \right)
	\]
	
	\[
	= \frac{\sigma_{\epsilon}^2}{\pi_1^2}.
	\]
	This is the required $V_\text{ridge}$, and we have our asymptotic distribution.

\end{proof}

\begin{theorem}[Asymptotic normality of the ridge estimator---Version 2]
	\label{thm-ridge-normality2}
	Again suppose $\lambda_n = o(\sqrt{n})$. Under Assumption \ref{assumption-stochastic}, the ridge estimator is asymptotically normal as well, and with the  same asymptotic variance, $V_\text{ridge}$.
\end{theorem}

\begin{proof}
	From a multivariate central limit theorem as in Lemma \ref{assumption-stochastic}, we again have asymptotic normality:
	\[
	\sqrt{n}\left(\left(\begin{array}{c}
	\hat{\text{Cov}}[Y_{i},Z_{i}]\\
	\hat{\text{Cov}}[D_{i},Z_{i}] 
	\end{array}\right)-\left(\begin{array}{c}
	\beta_{1}\pi_{1}\\
	\pi_{1}
	\end{array}\right)\right)\rightarrow\mathcal{N}\left(\left(\begin{array}{c}
	0\\
	0
	\end{array}\right),\Sigma \right).
	\]
	Note that the $\Sigma$ covariance matrix in this theorem is different from the one using Assumption \ref{assumption-constant}. Let us obtain it.\\
	
	We have $Y_i = \beta_1 \pi_1 Z_i + \beta_1 \eta_i + \epsilon_i$ and $D_i = \pi_1 Z_i + \eta_i.$ We ignore intercepts without loss of generality.\\
	Then,
	\[
	\text{Var}[Y_iZ_i] = E[Y_i^2 Z_i^2] - E[Y_i Z_i]^2
	\]
	\[
	= E\left[ (\beta_1 \pi_1 Z_i^2 + \beta_1 \eta_i Z_i + \epsilon_i Z_i)^2 \right] - (\beta_1 \pi_1)^2
	\]
	\[
	= E \left[ \beta_1^2 \pi_1^2 Z_i^4 + \beta_1^2 \eta_i^2 Z_i^2 + \epsilon_i^2 Z_i^2 + 2\beta_1^2 \pi_1 Z_i^3 \eta_i + 2\beta_1\eta_i \epsilon_i Z_i^2 + 2 \beta_1 \pi_1 Z_i^3 \epsilon_i  \right] - (\beta_1 \pi_1)^2
	\]
	
	\[
	= \beta_1^2 \pi_1^2 m_4 + \beta_1^2 \sigma_{\eta}^2 + \sigma_{\epsilon}^2 + 0 + 2\beta_1 \sigma_{\epsilon \eta}^2 + 0 - \beta_1^2 \pi_1^2
	\]
	
	\[
	= \beta_1^2 \pi_1^2 (m_4 - 1) + \beta_1^2 \sigma_{\eta}^2 + \sigma_{\epsilon}^2 + 2\beta_1 \sigma_{\epsilon \eta}^2.
	\]
	Next,
	\[
	\text{Var}[D_iZ_i] = \pi_1^2 (m_4 - 1) + \sigma_{\eta}^2
	\]
	from Lemma \ref{assumption-stochastic}. Finally,
	\[
	\text{Cov}[Y_iZ_i, D_iZ_i] = E[Y_iD_iZ_i^2] - E[Y_iZ_i]E[D_iZ_i]
	\]
	\[
	= E \left[ \left( \beta_1 \pi_1 Z_i + \beta_1 \eta_i + \epsilon_i \right) \left( \pi_1 Z_i + \eta_i \right)  Z_i^2 \right] - \beta_1 \pi_1^2
	\]
	\[
	= E \left[ \left( \beta_1 \pi_1^2 Z_i^2 + \beta_1 \pi_1 Z_i \eta_i + \beta_1 \pi_1 \eta_iZ_i + \beta_1 \eta_i^2 + \pi_1 \epsilon_i Z_i + \epsilon_i \eta_i \right) Z_i^2 \right] - \beta_1 \pi_1^2
	\]
	\[
	= \beta_1 \pi_1^2 m_4 + 0 + 0 + \beta_1 \sigma_{\eta}^2 + 0 + \sigma_{\epsilon \eta}^2 - \beta_1 \pi_1^2
	\]
	\[
	= \beta_1 \pi_1^2 (m_4 - 1) + \beta_1 \sigma_{\eta}^2 + \sigma_{\epsilon \eta}^2.
	\]
	From these it is clear that 
	\[
	\Sigma = \left( 
	\begin{array}{cc}
	\beta_1^2 \pi_1^2 (m_4 - 1) + \beta_1^2 \sigma_{\eta}^2 + \sigma_{\epsilon}^2 + 2\beta_1 \sigma_{\epsilon \eta}^2
	& \beta_1 \pi_1^2 (m_4 - 1) + \beta_1 \sigma_{\eta}^2 + \sigma_{\epsilon \eta}^2\\
	\beta_1 \pi_1^2 (m_4 - 1) + \beta_1 \sigma_{\eta}^2 + \sigma_{\epsilon \eta}^2 & 
	\pi_1^2 (m_4 - 1) + \sigma_{\eta}^2
	\end{array}
	\right).
	\]
	
	Using the multivariate delta method again as in the Theorem \ref{thm-ridge-normality1} on this new $\Delta$, we get the variance of the ridge estimator to be
	\[
	\frac{\sigma_{\epsilon}^2}{\pi_1^2},
	\]
	which is remarkably the same result.
	 
\end{proof}

We see that ridge IV recovers full efficiency in the ``large'' coefficient case. From another perspective, this result is uninteresting because ridge IV ``does nothing.'' The $\lambda_n$ rate being too slow, it yields the exact same distribution as the original 2SLS estimator without any penalization. To obtain a novel asymptotic distribution with ridge IV, we need a faster rate. We show that now.\\

To see what ridge IV is able to do, we set $\lambda_n = O(\sqrt{n})$. To be clear, in the previous theorems we use a sub-$\sqrt{n}$ rate, whereas now we use exactly a $\sqrt{n}$ rate for $\lambda_n$.  \cite{knight2000asymptotics} show that $\lambda_n = O(\sqrt{n})$ is necessary for unique $\sqrt{n}-$consistency of classic ridge regression. We show the same for ridge IV.

\begin{theorem}
	Let $\lambda_n = O(\sqrt{n})$, that is $\lambda_n / \sqrt{n} \rightarrow \lambda_0$, for some $\lambda_0 \geq 0$. Then 
	\[
	\sqrt{n} \left( \hat{\beta}_{\text{ridge}} - \beta_1 \right) \stackrel{d}{\rightarrow} \mathcal{N} \left( -\frac{\beta_1 \lambda_0}{\pi_1}, V_{\text{ridge}} \right).
	\]
\end{theorem}

\begin{proof}
	Let us look at the case with Assumption \ref{assumption-constant} first. The proof of Theorem \ref{thm-ridge-normality1} says
	\[
	\sqrt{n}\left(\left(\begin{array}{c}
	\hat{\text{Cov}}[Y_{i},Z_{i}]\\
	\hat{\text{Cov}}[D_{i},Z_{i}] + \frac{\lambda_n}{n}
	\end{array}\right)-\left(\begin{array}{c}
	\beta_{1}\pi_{1}\\
	\pi_{1}
	\end{array}\right)\right) \stackrel{d}{\rightarrow} \mathcal{N}\left(\left(\begin{array}{c}
	0\\
	\lambda_0
	\end{array}\right),\Sigma \right).
	\]
	The only thing that changes is the bias term with $\lambda_0$ because of the slower $O(\sqrt{n})$ rate of $\lambda_n$. Now take the ratio of the two elements of the vector and perform the delta method like in Theorem \ref{thm-ridge-normality1}. This gives us
	\[
	\sqrt{n} \left( \hat{\beta}_\text{ridge} - \beta_1 \right) \stackrel{d}{\rightarrow} \mathcal{N}\left(\nabla h(\beta_1\pi_1, \pi_1)^T \times \left(\begin{array}{c}
	0\\
	\lambda_0
	\end{array}\right), V_\text{ridge} \right),
	\]
	using the same notation as in the theorem \ref{thm-ridge-normality1}. This gives us our result.
	
\end{proof}

Under this new $\lambda_n$ regime, we see that we recover the same asymptotic distribution as the 2SLS estimator, but centered at the wrong mean! That is, the asymptotic bias is $-\frac{\beta_1 \lambda_0}{\pi_1}$, rather than 0. What then do we gain from the ridge approach?  The next section addresses this point.

\section{The Staiger-Stock critique}
\label{sec-staiger-stock}

Recall our main motivation for ridge IV is the weak instrument problem. To deal with it more explicitly, we take a local asymptotic framework used in \cite{staiger1997instrumental}. That is, we use
\[
\pi_1 = \frac{c}{\sqrt{n}}.
\]
This is to show that the strength of the first-stage is small, even relative to the sample size and the problem does not go away with bigger samples. Our next theorem shows the behavior of conventional 2SLS under this sense of weak instruments. 

\begin{theorem}
	\label{thm-2sls-staiger}
	Suppose our first stage is weak in the Staiger-Stock sense. That is, $\pi_1 = c/\sqrt{n}$. Then the 2SLS estimator is unstable and diverges. Specifically, $\hat{\beta}_\text{2SLS}$ converges to a Cauchy distribution, and so $\sqrt{n} \left( \hat{\beta}_\text{2SLS} - \beta_{1} \right)$ diverges.
\end{theorem}

\begin{proof}
	We operate under Assumption \ref{assumption-constant} here. We have 
	\[
	D_i = \pi_1 Z_i + \eta_i.
	\]
	Consider $D_i Z_i$. Since $\pi_1$ varies with sample size, $n$, we employ a triangular array argument here. $D_i Z_i$ has mean $\mu_{ni} = \pi_1 Z_i^2$ and variance $\sigma_{ni}^2 = Z_i^2 \sigma_{\eta}^2$.\\
	
	Let $T_n = \sum_{i=1}^n (D_i Z_i - \mu_{ni})$ and $s_n^2 = \text{Var}[T_n] = \sum_{i=1}^n \sigma_{ni}^2$. The Lindeberg-Feller Central Limit Theorem states
	\[
	\frac{T_n}{s_n} \stackrel{d}{\rightarrow} \mathcal{N}(0,1).
	\]
	Applying this to our setting, we get
	\[
	\frac{\sum_{i=1}^n \left( D_i Z_i - \pi_1 Z_i^2 \right) }{\sigma_{\eta} \sqrt{\sum_{i=1}^n Z_i^2}} \stackrel{d}{\rightarrow} \mathcal{N}(0,1).
	\]
	We have $\sum_{i=1}^n Z_i^2 / n \rightarrow 1$. So this can be rewritten as
	\[
	\sqrt{n} \left( \sum_{i=1}^n \frac{D_i Z_i - \pi_1 Z_i^2}{n} \right) \stackrel{d}{\rightarrow} \mathcal{N}(0, \sigma_{\eta}^2).
	\]
	Further, $\pi_1 = c/\sqrt{n}$. So this further simplifies to
	\[
	\sqrt{n} \sum_{i=1}^n \frac{D_i Z_i}{n} \stackrel{d}{\rightarrow} \mathcal{N}(c, \sigma_{\eta}^2).
	\]
	
	Applying the Lindeberg-Feller CLT similarly to the reduced form equation, we get
	\[
	\frac{\sum_{i=1}^n \left( Y_i Z_i - \beta_1 \pi_1 Z_i^2 \right) }{\sigma_{\text{red}} \sqrt{\sum_{i} Z_i^2}} \stackrel{d}{\rightarrow} \mathcal{N}(0,1).
	\]
	And by similar logic, this can be rewritten as
	\[
	\sqrt{n} \sum_{i=1}^n \frac{Y_i Z_i}{n} \stackrel{d}{\rightarrow} \mathcal{N}(c\beta_{1}, \sigma_{\text{red}}^2).
	\]
	
	Following the derivation of the covariance term in Theorem \ref{thm-ridge-normality1}, we can then were a joint normality result as follows:
	\[
	\frac{1}{\sqrt{n}} \sum_{i=1}^n \left( \begin{array}{c}
	Y_i Z_i - \beta_1 \pi_1 Z_i^2\\
	D_i Z_i - \pi_1 Z_i^2
	\end{array} \right) \stackrel{d}{\rightarrow} \mathcal{N} \left( \left( \begin{array}{c}
	0\\
	0
	\end{array} \right), 
	\left[\begin{array}{cc}
	\sigma_{\text{red}}^{2} & \beta_1 \sigma_{\eta}^2 + \sigma_{\epsilon \eta}^2\\
	\beta_1 \sigma_{\eta}^2 + \sigma_{\epsilon \eta}^2 & \sigma_{\eta}^{2}
	\end{array}\right] \right).
	\]	
	
	We called this covariance matrix $\Sigma$. With this notation, we can rewrite the joint normality as
	\[
	\sqrt{n} \left( \begin{array}{c}
	\hat{\text{Cov}}[Y_i, Z_i]\\
	\hat{\text{Cov}}[D_i, Z_i]
	\end{array} \right) \stackrel{d}{\rightarrow} \mathcal{N} \left( \left( \begin{array}{c}
	c\beta_1\\
	c
	\end{array} \right), \Sigma \right).
	\]
	
	Our 2SLS estimator is 
	\[
	\frac{\hat{\text{Cov}}[Y_i, Z_i]}{\hat{\text{Cov}}[D_i, Z_i]},
	\]
	so we can get its distribution by taking the ratio of the above joint normality distribution, which would result in a Cauchy distribution. Given that the ratio itself has a distribution, $\sqrt{n}$ times its difference from its mean would be unstable and diverges.
	
\end{proof}

How does ridge IV help with this? We demonstrate that in the next theorem, which is the most important result of the paper.
\begin{theorem}
	Let $\lambda_{n} = O(n)$. That is, $\lambda_{n} / n \rightarrow \lambda_0$ for some $ \lambda_0 \geq 0.$ Then, under Staiger-Stock asymptotics, we have
	\[
	\sqrt{n} \hat{\beta}_\text{ridge}  \stackrel{d}{\rightarrow} \mathcal{N} \left( \frac{c\beta_{1}}{\lambda_0}, \frac{\sigma_{\text{red}}^2}{\lambda_0^2} \right).
	\]
\end{theorem}
\begin{proof}
	The ridge IV estimator is 
	\[
	\frac{\hat{\text{Cov}}[Y_i, Z_i]}{\hat{\text{Cov}}[D_i, Z_i] + \frac{\lambda_{n}}{n}}.
	\]
	
	From Theorem \ref{thm-2sls-staiger}, we know that $\hat{\text{Cov}}[D_i, Z_i] = O_p(1/\sqrt{n})$. This implies
	\[
	\hat{\text{Cov}}[D_i, Z_i] \stackrel{p}{\rightarrow} 0,
	\]
	and from this it follows that 
	\[
	\hat{\text{Cov}}[D_i, Z_i] + \frac{\lambda_{n}}{n} \stackrel{p}{\rightarrow} \lambda_0.
	\]
	Also from Theorem \ref{thm-2sls-staiger}, we know 
	\[
	\sqrt{n} \hat{\text{Cov}}[Y_i, Z_i] \stackrel{d}{\rightarrow} \mathcal{N}(c\beta_{1}, \sigma_{\text{red}}^2).
	\]
	Then taking the ratio of the above two results using Slutsky's theorem, we get
	\[
	\sqrt{n} \left( \frac{\hat{\text{Cov}}[Y_i, Z_i]}{\hat{\text{Cov}}[D_i, Z_i] + \frac{\lambda_{n}}{n}} \right) \stackrel{d}{\rightarrow} \frac{1}{\lambda_0} \mathcal{N}(c\beta_{1}, \sigma_{\text{red}}^2).
	\]
	This is the required result.
	
\end{proof}

In the Staiger-Stock regime, ridge IV with an aggressive enough penalization scheme ($\lambda_{n} = O(n)$) massively lowers mean squared error of the estimate. In regimes where instruments are not weak, we can still use ridge IV with a more moderate penalization scheme ($\lambda_{n} = o(\sqrt{n})$), and lose nothing, although in this case, the choice of ridge IV over 2SLS is superfluous.

\section{Understanding ridge IV}
\label{sec-interpretation}

\subsection{Interpretation of $\lambda$ as the Lagrange multiplier in a constrained optimization problem}

In the simplest case, as in our model in (\ref{model}), ignoring intercepts, the standard 2SLS objective function is given by 
\[
\min_{\beta} \left( \sum_{i=1}^{n} Z_i (Y_i - D_i \beta) \right)^2.
\]

Then the ridge IV solution is given by the following objective function:
\begin{equation}
\label{ridge-iv-objective}
\min_{\beta} \left( \sum_{i=1}^{n} Z_i (Y_i - D_i \beta) \right)^2 + \gamma_n \beta^2.
\end{equation}

Clearly, setting penalization, $\lambda_n$ to zero recovers the original 2SLS estimator. What is the relation between $\lambda_n$ and the Lagrange multiplier $\gamma_n$ we see above? We address this in the next proposition.
\begin{proposition}[Objective function of ridge IV]
	\label{gamma-lambda-link}
	The objective function of ridge IV is as given in Equation \ref{ridge-iv-objective}. Further, there is a one-to-one relation between $\gamma_n$, the Lagrange multiplier in the objective function, and $\lambda_n$, the level of penalization in the ridge estimator, given by
	\[
	\gamma_n = \hat{\text{Cov}}[D_i, Z_i] \lambda_n.
	\]
\end{proposition}

\begin{proof}
	Let the ridge IV objective function be 
	\[
	\mathcal{L}(\beta) = \left( \sum_{i=1}^{n} Z_i (Y_i - D_i \beta) \right)^2 + \gamma_n \beta^2.
	\]
	
	This is a convex function, so to minimize it, set its partial derivative with respect to $\beta$ to zero. This gives us
	\[
	\frac{\partial}{\partial \beta} \mathcal{L} (\beta) = 2 \left( \sum_{i=1}^{n} Z_i (Y_i - D_i \beta) \right) \left( -\sum_{i=1}^n Z_iD_i \right) + 2\gamma_n \beta = 0.
	\]	
	That is,
	\[
	\left( \sum_{i=1}^{n} Z_i D_i \beta - \sum_{i=1}^n Z_i Y_i \right) \left( \sum_{i=1}^n Z_iD_i \right) + \gamma_n \beta = 0.
	\]
	This gives us
	\[
	\left(\left( \sum_i Z_iD_i \right)^2 + \gamma_n \right) \beta = \left( \sum_i Z_iD_i \right) \left( \sum_i Z_iY_i \right),
	\]
	or
	\[
	\hat{\beta} = \frac{\sum_i Z_iY_i}{\sum_i Z_iD_i + \frac{\gamma_n}{\sum_i Z_iD_i}}.
	\]
	
	Comparing this form with the definition of the ridge IV estimator in Equation (\ref{eqn-ridge-defn}), we know
	\[
	\frac{\lambda_n}{n} = \frac{\gamma_n}{\sum_i Z_iD_i}.
	\]
	Rearranging the equation gives us the desired result.
		
\end{proof}



\section{Results}
\label{sec-results}

We look at a simulation design where we are interested in seeing how the MSE metric varies with varying levels of aggressiveness in the ridge penalty, $\lambda$. The linear IV model for this simulation is
\[
Y=2.83+\text{effect}\times D+\epsilon,
\]
\[
D=-0.346+0.072\times Z-0.67\times\epsilon+\eta.
\]
This corresponds to our model in equation \ref{model}, with coefficients set to $\beta_0=2.83, \pi_0 = -0.346, \pi_1 = 0.072.$ The coefficients and sample sizes are chosen to match a study from the AER \citep{hornung2014immigration}. $\epsilon,\eta,Z$ are some independent normals.\\ 

In the first set of results, we allow the first stage coefficient size, $]pi_1$ to vary. 
Our simulation study is as follows:
\begin{enumerate}
	\item Pick a first stage coefficient from 0 to 1. 
	\item Simulate 10,000 datasets of N=150 each.
	\item Compute MSE of that estimated coefficient for $\beta_1$ (including intercept).
\end{enumerate}

\begin{figure}[H]
	\centering
	\subfloat[$\lambda = 0$ (Regular 2SLS)]{{\includegraphics[width=5cm]{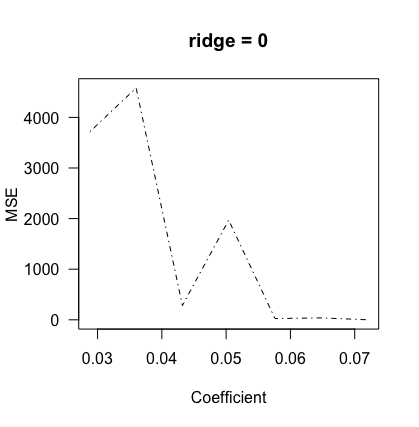} }}
	\qquad
	\subfloat[$\lambda = 4.0$ (Moderate regularization)]{{\includegraphics[width=5cm]{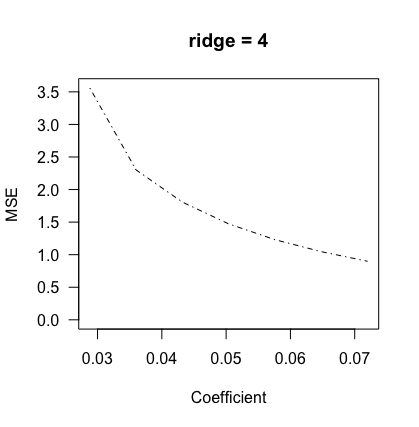} }}%
	\qquad
	\subfloat[$\lambda = 10$ (Aggressive regularization)]{{\includegraphics[width=5cm]{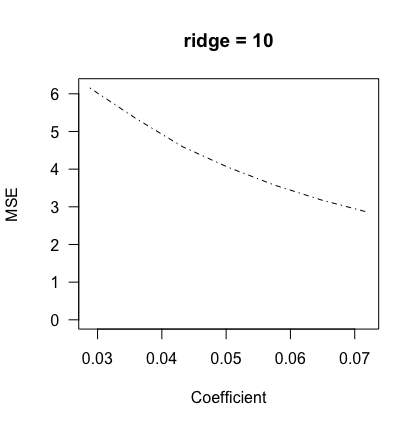} }}%
	\caption{MSE of the IV estimator by the level of ridge penalization, $\lambda$. Figure (a) shows the regular 2SLS case, which we compare against. It is extremely unstable when the instrument is weak, i.e. when we are closer to 0 on the x-axis. When using a moderate level of regularization as in (b), we see that MSE is reduced by three orders of magnitude. Figure (c) sounds the alarm against too much regularization because then bias can become a dominant force to vie with and MSE starts to pick up again.}%
	\label{fig:mse-ridgeiv}%
\end{figure}

In the second set of results, we allow effect size, $\beta_1$ to vary. This is to show that for a given first stage strength, we may still have some use for regularization for very small coefficients. This simulation study is as follows:
\begin{enumerate}
	\item Pick an effect size, i.e. $\beta_1$ from 0 to 3.475.
	\item Simulate 10,000 datasets of N=150 each.
	\item Compute MSE of that estimated coefficient (including intercept).
\end{enumerate}

\begin{figure}[H]
	\centering
	\subfloat[$\lambda = 0$ (Regular 2SLS)]{{\includegraphics[width=5cm]{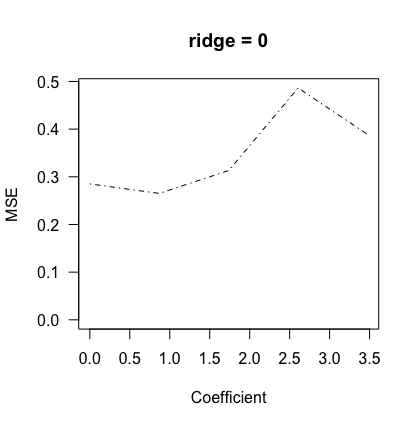} }}
	\qquad
	\subfloat[$\lambda = 0.8$ (Moderate regularization)]{{\includegraphics[width=5cm]{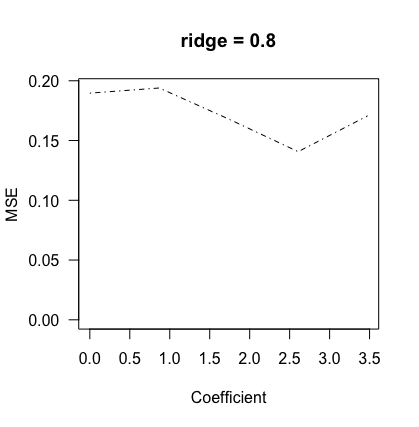} }}%
	\qquad
	\subfloat[$\lambda = 3.0$ (Aggressive regularization)]{{\includegraphics[width=5cm]{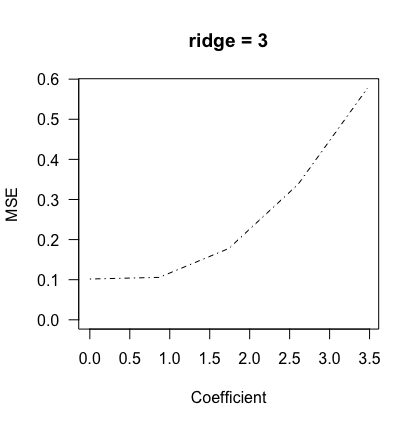} }}%
	\caption{MSE of the IV estimator by the level of ridge penalization, $\lambda$. Figure (a) shows the regular 2SLS case, which we compare against. When using a moderate level of regularization as in (b), we see that MSE may be halved. However, figure (c) sounds the alarm against too much regularization because then bias can become a dominant force to vie with, particularly when effect sizes are large.}%
	\label{fig:mse-ridgeiv}%
\end{figure}

We show three cases here. The case of zero regularization is the classic 2SLS estimator, which has a certain level of mean squared error, which remains high even when the effect size being studied is small. 


\section{Conclusion} \label{sec-disc}
	In this paper, we introduced a novel estimator, called ``ridge IV.'' We motivated it as the solution to the GMM objective function for instrumental variable regression with an additional $L_2$ penalty. In the theoretical case with ``large'' coefficients, we showed that ridge IV does not hurt our estimation while in the weak first stage case, we showed that it leads to strong improvements in the mean squared error of the estimand. We then validated the theory using simulations inspired by data designs of papers in the American Economic Review.\\
	
	While this paper is primarily a theoretical contribution to the literature, we outline several avenues for further research. First, it would be helpful to provide a method for tuning the $\lambda$ parameter for a given dataset. Our results operated at the abstract level of the big-$O$ notation, but for practical use, more information is needed. We would also like to see results on how exactly to perform inference with ridge IV. Explicit demonstration of Type 1 error control, for instance, would be very useful. \\ 
	
	We would like to tie back the results of ridge IV and interpret them in the context of the problem of IV sensitivity to outliers, which is related to instability. We conjecture that ridge IV under an appropriate penalization scheme can address this as well. We would also like to expand the results provided in this paper in the general case of over-identifying instruments. A natural extension of our estimator is given in Equation (\ref{defn-overidentified-ridgeIV}).\\
	
\bibliographystyle{ecta}
\bibliography{references}

\end{document}